\documentclass[11pt]{article}%
\usepackage{amsmath}
\usepackage{amsfonts}
\usepackage{amssymb}
\usepackage{graphicx}%
\newtheorem{theorem}{Theorem}

\newtheorem{corollary}[theorem]{Corollary}

\newtheorem{definition}[theorem]{Definition}
\newtheorem{example}[theorem]{Example}

\newtheorem{lemma}[theorem]{Lemma}

\newenvironment{proof}[1][Proof]{\noindent\textbf{#1.} }{\ \rule{0.5em}{0.5em}}
\graphicspath{%
    {/}% inserted by PCTeX
}
\begin{document}
\centerline{{\large \textbf{Factorizations and Reductions of Order in Quadratic and other} }}
\vspace{1ex}
\centerline{{\large \textbf{Non-recursive Higher Order Difference Equations} }}

\vspace{4ex}

\centerline{H. SEDAGHAT*} \footnotetext{\noindent*Department of Mathematics,
Virginia Commonwealth University, Richmond, Virginia 23284-2014, USA, Email:
hsedagha@vcu.edu}

\vspace{3ex}

\begin{quote}
\noindent{\small \textbf{Abstract.} A higher order difference equation may be
generally defined in an arbitrary nonempty set }${\small S}$ {\small as:}%
\[
f_{n}(x_{n},x_{n-1},\ldots,x_{n-k})=g_{n}(x_{n},x_{n-1},\ldots,x_{n-k})
\]
{\small where }${\small f}_{{\small n}}{\small ,g}_{{\small n}}{\small :S}%
^{k+1}{\small \rightarrow S}$ {\small are given functions for }%
${\small n=1,2,\ldots}$ {\small and }${\small k}$ {\small is a positive
integer. We present conditions that imply the above equation can be factored
into an equivalent pair of lower order difference equations using possible
form symmetries (order-reducing changes of variables). These results extend
and generalize semiconjugate factorizations of recursive difference equations
on groups. We apply some of this theory to obtain new factorization results
for the important class of quadratic difference equations on algebraic
fields:}%
\[
\sum_{i=0}^{k}\sum_{j=i}^{k}a_{i,j,n}x_{n-i}x_{n-j}+\sum_{j=0}^{k}%
b_{j,n}x_{n-j}+c_{n}=0.
\]

{\small We also discuss the nontrivial issue of the existence of solutions 
for quadratic equations.}
\end{quote}

\section{Introduction}

Let $S$ be a nonempty set and $\{f_{n}\}_{n=1}^{\infty},\{g_{n}\}_{n=1}^{\infty}$ 
be sequences of functions $f_{n},g_{n}:S^{k+1}\rightarrow S$ where $k$ is a positive 
integer. We call the equation
\begin{equation}
f_{n}(x_{n},x_{n-1},\ldots,x_{n-k})=g_{n}(x_{n},x_{n-1},\ldots,x_{n-k}),\quad
n=1,2,3,\ldots\label{nrdek}%
\end{equation}
a \textit{non-recursive difference equation of order} $k$ in the set $S$ if
$f_{n}(u_{0},\ldots,u_{k})$ is not constant in $u_{0}$ and $g_{n}(u_{0}%
,\ldots,u_{k})$ is not constant in $u_{k}$ for all $n\geq1.$ A
\textit{(forward) solution} of (\ref{nrdek}) is a sequence $\{x_{n}%
\}_{n=1}^{\infty}$ in $S$ that makes (\ref{nrdek}) true, given a set of
initial values $x_{-j}$, $j=0,1,\ldots,k-1.$ The \textit{recursive},
non-autonomous equation of order $k,$ i.e.,
\begin{equation}
x_{n}=\psi_{n}(x_{n-1},\ldots,x_{n-k}) \label{rgen}%
\end{equation}
is a special case of (\ref{nrdek}) with%
\begin{align*}
f_{n}(u_{0},u_{1},\ldots,u_{k})  &  =u_{0}\\
g_{n}(u_{0},u_{1},\ldots,u_{k})  &  =\psi_{n}(u_{1},\ldots,u_{k})
\end{align*}
for all $n$ and all $u_{0},u_{1},\ldots,u_{k}\in S.$

An example of non-recursive equations in familiar settings is the
following on the set $\mathbb{R}$ of all real numbers:%
\begin{equation}
|x_{n}|=a|x_{n-1}-x_{n-2}|,\quad0<a<1 \label{nrsm}%
\end{equation}

This states that the \textit{magnitude} of a quantity $x_{n}$ at time $n$ is a
fraction of the difference between its values in the two immediately preceding
times; however, we cannot determine the \textit{sign} of $x_{n}$ from
(\ref{nrsm}). As a possible physical interpretation of (\ref{nrsm}) imagine a
node in a circuit that in every second $n$ fires a pulse $x_{n}$ that may go
either to the right (if $x_{n}>0$) or to the left (if $x_{n}<0$) but the
amplitude $|x_{n}|$ of the pulse obeys Eq.(\ref{nrsm}). With regard to the
variety of solutions for (\ref{nrsm}) we note that the direction of each pulse
is entirely unpredictable, regardless of the directions of previous pulses
emitted by the node; hence a large number of solutions are possible for
(\ref{nrsm}). We discuss this equation in greater detail in the next section.

For recursive equations such as (\ref{rgen}) on groups, recent studies such as
\cite{ks}, \cite{expow}, \cite{hstdfs}, \cite{hsijpam}, \cite{hd1}, 
\cite{hsinvcrt}, \cite{kyoto}, \cite{hsarx}, show that possible form symmetries 
(i.e., order-reducing coordinate
transformations or changes of variables) and the associated semiconjugate
relations may be used to break down the equation into a triangular system \cite{smit} of
lower order difference equations whose orders add up to the order of
(\ref{rgen}). But in general it is not possible to write (\ref{nrdek}) in the
recursive form (\ref{rgen}) so the question arises as to whether the notions
of form symmetry and reduction of order can be extended to the more general
non-recursive context.

In this paper, we show that for Eq.(\ref{nrdek}) basic concepts such as form
symmetry and factorization into factor and cofactor pairs of equations can
still be defined as before, even without a semiconjugate relation. A concept
that is similar to semiconjugacy but which does not require the unfolding map
is sufficient for defining form symmetries and deriving the lower order factor
and cofactor equations. Using this idea we extend previously established
theory of factorization and reduction of order to a much larger class of
difference equations then previously studied. In particular, we apply this
extended theory to the important class of quadratic difference equations.
%
%%%%%%%%%%%%%%%%%%%%%%%%%%%%%%%%%%%%%%%%%%%%%%%%%%%
%%%%%%%%%%%%%%%%%%%%%%%%%%%%%%%%%%%%%%%%%%%%%%%%%%%
%
\section{Factorizations of non-recursive equations}

Equation (\ref{nrdek}) generalizes the recursive equation
(\ref{rgen}) in a different direction than the customary one, namely, through
unfolding the recursive equation of order $k$ to a special vector map of the
$k$-dimensional \textit{state-space} $S^{k}$. Nevertheless, it is convenient
to define as \textit{states} the points $(x_{n},x_{n-1},\ldots,x_{n-k+1})\in
S^{k}$ or some invariant subset of it that contains the \textit{orbit}

$$\{(x_{n},x_{n-1},\ldots,x_{n-k+1})\}_{n=1}^{\infty}$$ 

\noindent of every solution $\{x_{n}\}_{n=1}^{\infty}$ of (\ref{nrdek}). We may designate 
the point $(x_{0},\ldots,x_{-k+1})$ corresponding to $n=0$ on each orbit as the
\textquotedblleft initial point\textquotedblright\ of that orbit; in contrast
to recursive equations however, there may be many distinct orbits
having the same initial point.

Analyzing the solutions of a non-recursive difference equation such as
(\ref{nrdek}) is generally more difficult than analyzing the solutions of
recursive equations. Unlike the recursive case, even the existence of
solutions for (\ref{nrdek}) in a particular set is not guaranteed. But
studying the form symmetries and reduction of order in non-recursive equations
is worth the effort. The greater generality of these equations not only leads
to the resolution of a wider class of problems, but it also provides for
increased flexibility in handling \textit{recursive} equations.

Before beginning the formal study of factorizations of non-recursive
equations, let us consider an illustrative example that highlights several
issues pertaining to such equations.

\begin{example}
\label{nr3ex}Let $\{a_{n}\}_{n=1}^{\infty}$ be any sequence of non-negative real numbers 
and consider the following third-order difference equation on $\mathbb{R}$:%
\begin{equation}
|x_{n}+x_{n-1}|=a_{n}|x_{n-1}-x_{n-3}|. \label{nrsm3}%
\end{equation}
By adding and subtracting $x_{n-2}$ inside the absolute value on the right
hand side of (\ref{nrsm3}) we find that%
\begin{equation}
|x_{n}+x_{n-1}|=a_{n}|x_{n-1}+x_{n-2}-(x_{n-2}+x_{n-3})| \label{nrsm2}%
\end{equation}
The substitution
\begin{equation}
t_{n}=x_{n}+x_{n-1} \label{nrsm0}%
\end{equation}
in (\ref{nrsm2}) results in the second-order difference equation%
\begin{equation}
|t_{n}|=a_{n}|t_{n-1}-t_{n-2}| \label{nrsm1}%
\end{equation}
that is related to (\ref{nrsm3}) via (\ref{nrsm0}). Eq.(\ref{nrsm1}) is
analogous to a factor equation for (\ref{nrsm3}) while (\ref{nrsm0}), written
as%
\begin{equation}
x_{n}=t_{n}-x_{n-1} \label{nrsm1c}%
\end{equation}
is analogous to a (recursive) cofactor equation. Also the substitution
(\ref{nrsm0}) is analogous to an order-reducing form symmetry; see, e.g.,
\cite{hstdfs}, \cite{hsijpam}, \cite{hd1} or \cite{hsarx}.

The factor equation (\ref{nrsm1}) is of course not recursive; it is a
generalization of (\ref{nrsm}) in the introduction above. If $\{\beta
_{n}\}_{n=1}^{\infty}$ is any fixed but arbitrary binary sequence taking
values in $\{-1,1\}$ then every real solution $\{s_{n}\}_{n=1}^{\infty}$ of
the recursive equation%
\begin{equation}
s_{n}=\beta_{n}a_{n}|s_{n-1}-s_{n-2}| \label{nrsm6}%
\end{equation}
is also a solution of (\ref{nrsm1}). This follows upon taking the absolute
value to see that $\{s_{n}\}_{n=1}^{\infty}$ satisfies Eq.(\ref{nrsm1}). The
single non-recursive equation (\ref{nrsm1}) has as many solutions as can be
generated by each member of the uncountably infinite class of equations
(\ref{nrsm6}) put together. Numerical simulations and other calculations
indicate a wide variety of different solutions for (\ref{nrsm6}) with
different choices of $\{\beta_{n}\}_{n=1}^{\infty}$ and clearly no less is
true about (\ref{nrsm3}).

These facts remain true in the special case mentioned in the Introduction,
i.e., equation (\ref{nrsm}) in which the sequence $a_{n}=a$ is constant. By
way of comparison first consider the case where $\beta_{n}=1$ is constant for
all $n.$ Then each solution $\{s_{n}\}_{n=1}^{\infty}$ of the recursive
difference equation%
\[
s_{n}=a|s_{n-1}-s_{n-2}|,\quad0<a<1
\]
is uniquely defined by an initial point $(s_{0},s_{-1})\in\mathbb{R}^{2}$ and
\begin{equation}
\lim_{n\rightarrow\infty}s_{n}=0. \label{lsn}%
\end{equation}

This claim is proved as follows. Without loss of generality assume that
$s_{0}\geq0$ and define%
\[
\mu=\max\{s_{0},s_{1}\}\geq0.
\]
Then%
\begin{align*}
s_{2}  &  =a|s_{1}-s_{0}|\leq a\max\{s_{0},s_{1}\}\leq a\mu,\\
s_{3}  &  =a|s_{2}-s_{1}|\leq a\max\{s_{2},s_{1}\}\leq a\mu.
\end{align*}
Continuing,
\begin{align*}
s_{4}  &  =a|s_{3}-s_{2}|\leq a\max\{s_{3},s_{2}\}\leq a^{2}\mu,\\
s_{5}  &  =a|s_{4}-s_{3}|\leq a\max\{s_{4},s_{3}\}\leq a^{2}\mu.
\end{align*}
This reasoning by induction yields%
\[
s_{2n},s_{2n+1}\leq a^{n}\mu,\quad\text{for all }n
\]
which proves (\ref{lsn}).

By contrast, if $\{\beta_{n}\}_{n=1}^{\infty}$ is not a constant sequence then
complicated solutions may occur for (\ref{nrsm}). The computer generated
diagram in Figure \ref{fig:nr-absval-chs} shows a part of the solution of the 
difference equation
\[
s_{n}=a\beta_{n}|s_{n-1}-s_{n-2}|
\]

\noindent with parameter values:%
\begin{align*}
a  &  =0.8,\quad s_{-1}=s_{0}=1\\
\beta_{n}  &  =\left\{
\begin{array}
[c]{c}%
-1\quad\text{if }r_{n}<0.45\\
1\quad\text{if }r_{n}\geq0.45
\end{array}
\right. \\
r_{n}  &  =3.75r_{n-1}(1-r_{n-1}),\quad r_{0}=0.4.
\end{align*}

%********* FIGURE **********

\begin{figure}[tbp] % float placement: (h)ere, page (t)op, page (b)ottom, other (p)age
  \centering
  % file name: C:/HS-Math/arXiv/nr-absval-chs.bmp
  \includegraphics[bb=0 0 515 259,width=3.95in,height=1.99in,keepaspectratio]{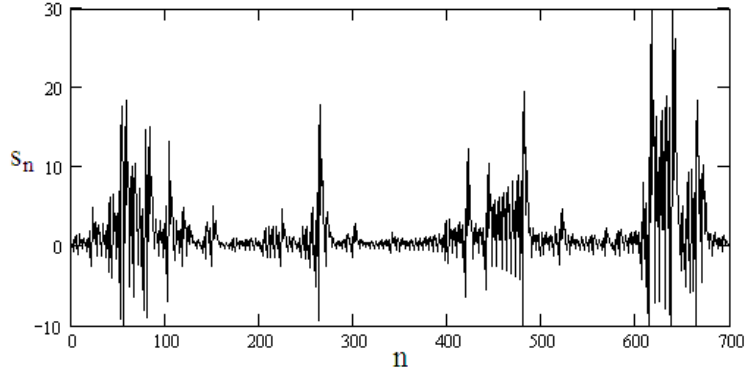}
  \caption{A complex solution of equation (\ref{nrsm}) highlighting the unpredictability of sign changes}
  \label{fig:nr-absval-chs}
\end{figure}

The complex behavior seen in Figure \ref{fig:nr-absval-chs} is indicative of the stochastic nature
of the pulse direction in the physical interpretation in the Introduction.
Although $\{\beta_{n}\}_{n=1}^{\infty}$ does not appear explicitly in
(\ref{nrsm}) the unpredictability of the sign of $x_{n}$ is a way of
interpreting the arbitrary nature of the binary sequence.

Finally, we observe that equation (\ref{nrsm1}) which has order two, also
admits a reduction of order as follows. If $\{t_{n}\}$ is a solution of
(\ref{nrsm1}) that is never zero for all $n$ then we may divide both sides of
(\ref{nrsm1}) by $|t_{n-1}|$ to get%
\[
\left\vert \frac{t_{n}}{t_{n-1}}\right\vert =a_{n}\left\vert 1-\frac{t_{n-2}%
}{t_{n-1}}\right\vert
\]
where the substitution%
\[
r_{n}=\frac{t_{n}}{t_{n-1}}%
\]
(analogous to the inversion form symmetry; see \cite{hsarx}) yields the
first-order difference equation%
\begin{equation}
\left\vert r_{n}\right\vert =a_{n}\left\vert 1-\frac{1}{r_{n-1}}\right\vert .
\label{nrsm4}%
\end{equation}
Eq. (\ref{nrsm4}) is related to (\ref{nrsm1}) via the (recursive) equation%
\begin{equation}
t_{n}=r_{n}t_{n-1}. \label{nrsm4c}%
\end{equation}

\end{example}

The factorizations and corresponding reductions in order given by equations
(\ref{nrsm1}), (\ref{nrsm1c}), (\ref{nrsm4}) and (\ref{nrsm4c}) are among the
types discussed below along with a variety of other possibilities. For
additional examples and further details see \cite{hsfsorbk}.
%
%%%%%%%%%%%%%%%%%%%%%%%%%%%%%%%%%%%%%%%%%%%%%%%%%%%
%
\subsection{Form symmetries, factors and cofactors}

In this section we define the concepts of order-reducing form symmetry and the
associated factorization for Eq.(\ref{nrdek}) on an arbitrary set $S.$ In
analogy to semiconjugate factorizations, we seek a decomposition of
Eq.(\ref{nrdek}) into a pair of difference equations of lower orders. A factor
equation of type%
\begin{equation}
\phi_{n}(t_{n},t_{n-1},\ldots,t_{n-m})=\psi_{n}(t_{n},t_{n-1},\ldots
,t_{n-m}),\quad1\leq m\leq k-1 \label{fnrde}%
\end{equation}
may be derived from (\ref{nrdek}) where $\phi_{n},\psi_{n}:S^{m+1}\rightarrow
S$ for all $n$ if there is a sequence of mappings $H_{n}:S^{k+1}\rightarrow
S^{m+1}$ such that
\begin{equation}
f_{n}=\phi_{n}\circ H_{n}\text{ and }g_{n}=\psi_{n}\circ H_{n} \label{nrr}%
\end{equation}
for all $n\geq1.$ If we denote%
\[
H_{n}(u_{0},\ldots,u_{k})=[h_{0,n}(u_{0},\ldots,u_{k}),h_{1,n}(u_{0}%
,\ldots,u_{k}),\ldots,h_{m,n}(u_{0},\ldots,u_{k})]
\]
then for each solution $\{x_{n}\}$ of Eq.(\ref{nrdek})%
\begin{gather*}
\phi_{n}(h_{0,n}(x_{n},\ldots,x_{n-k}),h_{1,n}(x_{n},\ldots,x_{n-k}%
),\ldots,h_{m,n}(x_{n},\ldots,x_{n-k}))=\\
\psi_{n}(h_{0,n}(x_{n},\ldots,x_{n-k}),h_{1,n}(x_{n},\ldots,x_{n-k}%
),\ldots,h_{m,n}(x_{n},\ldots,x_{n-k}))
\end{gather*}

In order for a sequence $\{t_{n}\}$ in $S$ defined by the substitution%
\[
t_{n}=h_{0,n}(x_{n},\ldots,x_{n-k})
\]
to be a solution of (\ref{fnrde}), the functions $H_{n}$ must have a special
form that is defined next.

\begin{definition}
\label{nrfsdef}A sequence of functions $\{H_{n}\}$ is an
\textbf{order-reducing form symmetry} of Eq.(\ref{nrdek}) on a nonempty set
$S$ if there is an integer $m$, $1\leq m<k,$ and sequences of functions
$\phi_{n},\psi_{n}:S^{m+1}\rightarrow S$ and $h_{n}:S^{k-m+1}\rightarrow S$
such that
\begin{equation}
H_{n}(u_{0},\ldots,u_{k})=[h_{n}(u_{0},\ldots,u_{k-m}),h_{n-1}(u_{1}%
,\ldots,u_{k-m+1}),\ldots,h_{n-m}(u_{m},\ldots,u_{k})] \label{nrfs}%
\end{equation}
and the sequences $\{\phi_{n}\}$, $\{\psi_{n}\},\{f_{n}\},\{g_{n}\}$ and
$\{H_{n}\}$ satisfy the relations (\ref{nrr}) for all $n\geq1$.
\end{definition}

If $S=(G,\ast)$ is a group then Definition \ref{nrfsdef} generalizes the
notion of recursive form symmetry in \cite{hstdfs} where the components of
$H_{n}$ are defined as%
\[
h_{n-j}(u_{j},\ldots,u_{k-m+j})=u_{j}\ast\widetilde{h}_{n-j}(u_{j+1}%
,\ldots,u_{k-m+j}),\quad j=0,1,\ldots,m.
\]

We have the following basic factorization theorem for non-recursive difference equations.

\begin{theorem}
\label{nrfacthm}Assume that Eq.(\ref{nrdek}) has an order-reducing form
symmetry $\{H_{n}\}$ defined by (\ref{nrfs}) on a nonempty set $S.$ Then the
difference equation (\ref{nrdek})\ has a factorization into an equivalent
system of factor and cofactor equations
\begin{align}
\phi_{n}(t_{n},\ldots,t_{n-m})  &  =\psi_{n}(t_{n},\ldots,t_{n-m}%
)\label{nrf}\\
h_{n}(x_{n},\ldots,x_{n-k+m})  &  =t_{n} \label{nrcf}%
\end{align}
whose orders $m$\ and $k-m$\ respectively, add up to the order of (\ref{nrdek}).
\end{theorem}

\begin{proof}
To show the equivalence, we show that for each solution $\{x_{n}\}$ of
(\ref{nrdek}) there is a solution $\{(t_{n},y_{n})\}$ of the system of
equations (\ref{nrf}), (\ref{nrcf}) such that $y_{n}=x_{n}$ for all $n\geq1$
and conversely, for each solution $\{(t_{n},y_{n})\}$ of the system of
equations (\ref{nrf}) and (\ref{nrcf}) the sequence $\{y_{n}\}$ is a solution
of (\ref{nrdek}).

First assume that $\{x_{n}\}$ is a solution of Eq.(\ref{nrdek}) through a
given initial point $(x_{0},x_{-1},\ldots,x_{-k+1})\in G^{k+1}.$ Define the
sequence $\{t_{n}\}$ in $S$ as in (\ref{nrcf}) for $n\geq-m+1$ so that by
(\ref{nrr})%
\begin{align*}
\phi_{n}(t_{n},\ldots,t_{n-m})  &  =\phi_{n}(h_{n}(x_{n},\ldots,x_{n-k+m}%
),\ldots,h_{n-m}(x_{n-m},\ldots,x_{n-k}))\\
&  =\phi_{n}(H_{n}(x_{n},\ldots,x_{n-k}))\\
&  =f_{n}(x_{n},\ldots,x_{n-k})\\
&  =g_{n}(x_{n},\ldots,x_{n-k})\\
&  =\psi_{n}(H_{n}(x_{n},\ldots,x_{n-k}))\\
&  =\psi_{n}(t_{n},\ldots,t_{n-m}).
\end{align*}

It follows that $\{t_{n}\}$ is a solution of (\ref{nrf}). Further, if
$y_{n}=x_{n}$ for $n\geq-k+m$ then by the definition of $t_{n}$, $\{y_{n}\}$
is a solution of (\ref{nrcf}).

Conversely, let $\{(t_{n},y_{n})\}$ be a solution of the factor-cofactor
system (\ref{nrf}), (\ref{nrcf}) with initial values
\[
t_{0},\ldots,t_{-m+1},y_{-m},\ldots y_{-k+1}\in G.
\]

We note that $y_{0},y_{-1},\ldots,y_{-m+1}$ satisfy the equations
\[
h_{j}(y_{j},\ldots,y_{j-k+m})=t_{j},\quad j=0,-1,\ldots,-m+1.
\]

Now for $n\geq1$, (\ref{nrr}) implies%
\begin{align*}
f_{n}(y_{n},\ldots,y_{n-k})  &  =\phi_{n}(H_{n}(y_{n},\ldots,y_{n-k}))\\
&  =\phi_{n}(h_{n}(y_{n},\ldots,y_{n-k+m}),\ldots,h_{n-m}(y_{n-m}%
,\ldots,y_{n-k}))\\
&  =\phi_{n}(t_{n},\ldots,t_{n-m})\\
&  =\psi_{n}(t_{n},\ldots,t_{n-m})\\
&  =\psi_{n}(H_{n}(x_{n},\ldots,x_{n-k}))\\
&  =g_{n}(x_{n},\ldots,x_{n-k}).
\end{align*}

Therefore, $\{y_{n}\}$ is a solution of (\ref{nrdek}).
\end{proof}

The concept of \textit{order reduction types} for non-recursive difference
equations can be defined similarly to recursive equations as in prior studies
and is not repeated here.
%
%%%%%%%%%%%%%%%%%%%%%%%%%%%%%%%%%%%%%%%%%%%%%%%%%%%
%
\subsection{Semi-invertible map criterion}

A group structure is necessary for obtaining certain results such as an
extension of the useful invertible map criterion in \cite{hstdfs} and
\cite{hsinvcrt} to non-recursive equations. In this section we assume that
$S=(G,\ast)$ is a non-trivial group with the goal of obtaining an extension of
the invertible map criterion to the non-recursive equation (\ref{nrde}).
For a discussion of some issues pertaining to difference and differential 
equations on algebraic rings see \cite{bert}.

Denoting the identity of $G$ by $\iota$, the difference equation (\ref{nrdek})
can be written in the equivalent form%
\begin{equation}
E_{n}(x_{n},x_{n-1},\ldots,x_{n-k})=\iota,\quad k\geq2,\ n=1,2,3,\ldots
\label{nrde}%
\end{equation}
where $E_{n}=f_{n}\ast\lbrack g_{n}]^{-1}$ (the brackets indicate group
inversion). A type-($k-1,1$) reduction of Eq.(\ref{nrde}) is characterized by
the following factorization%
\begin{align}
\phi_{n}(t_{n},\ldots,t_{n-k+1})  &  =\iota\label{nrk1f}\\
h_{n}(x_{n},x_{n-1})  &  =t_{n} \label{nrk1cf}%
\end{align}
in which the cofactor equation has order one. This system occurs if the
following is a form symmetry of (\ref{nrde}):%
\begin{equation}
H_{n}(u_{0},\ldots,u_{k})=[h_{n}(u_{0},u_{1}),h_{n-1}(u_{1},u_{2}%
),\ldots,h_{n-k+1}(u_{k-1},u_{k})]. \label{nrk1fs}%
\end{equation}

A specific example of the above factorization is the system of equations
(\ref{nrsm1}) and (\ref{nrsm1c}) in Example \ref{nr3ex} which factor
Eq.(\ref{nrsm3}). More examples are discussed below.

The following type of coordinate function $h_{n}$ is of particular interest in
this section.

\begin{definition}
A coordinate function $h:G^{2}\rightarrow G$ on a non-trivial group $G$ is
\textbf{separable} if%
\[
h(u,v)=\mu(u)\ast\theta(v)
\]
for given self-maps $\mu,\theta$ of $G$ into itself. A separable $h$ is
\textbf{right semi-invertible} if $\theta$ is a bijection and \textbf{left
semi-invertible} if $\mu$ is a bijection. If both $\theta$ and $\mu$ are
bijections then $h$ is \textbf{semi-invertible}. A form symmetry $\{H_{n}\}$
is (right, left) semi-invertible if the coordinate function $h_{n}$ is (right,
left) semi-invertible for every $n$.
\end{definition}

Note that a semi-invertible $h$ is \textit{not} a bijection in general; for
instance, consider $h(u,v)=u-v$ where $G$ is the group of all real numbers
under addition.

Clearly, functions of type $u\ast\widetilde{h}(v)$ where $\widetilde{h}$ is a
bijection are semi-invertible functions. Therefore, semi-invertible functions
generalize the types of maps discussed previously in \cite{hstdfs}, \cite{hsijpam}, 
\cite{hsinvcrt} and \cite{hsarx}. The next theorem shows that the invertible
map criterion discussed in these references can be extended to all right
semi-invertible form symmetries.

\begin{theorem}
\label{seminv}(Semi-invertible map criterion) Assume that $h_{n}(u,v)=\mu
_{n}(u)\ast\theta_{n}(v)$ is a sequence of right semi-invertible functions
with bijections $\theta_{n}$ of a group $G.$ For arbitrary $u_{0},v_{1}%
,\ldots,v_{k}\in G$ define $\zeta_{0,n}\equiv u_{0}$ and for $j=1,\ldots,k$
\begin{equation}
\zeta_{j,n}(u_{0},v_{1},\ldots,v_{j})=\theta_{n-j+1}^{-1}([\mu_{n-j+1}%
(\zeta_{j-1,n}(u_{0},v_{1},\ldots,v_{j-1}))]^{-1}\ast v_{j}) \label{nrzetaj}%
\end{equation}
with the usual distinction observed between map inversion and group inversion.
Then Eq.(\ref{nrde}) has the form symmetry (\ref{nrk1fs}) and the associated
factorization into equations (\ref{nrk1f}) and (\ref{nrk1cf}) if and only if
the following quantity%
\begin{equation}
E_{n}(u_{0},\zeta_{1,n}(u_{0},v_{1}),\ldots,\zeta_{k,n}(u_{0},v_{1}%
,\ldots,v_{k})) \label{Enzeta}%
\end{equation}
is independent of $u_{0}$ for all $n.$ In this case, the factor functions
$\phi_{n}$ are given by%
\begin{equation}
\phi_{n}(v_{1},\ldots,v_{k})=E_{n}(u_{0},\zeta_{1,n}(u_{0},v_{1}),\ldots
,\zeta_{k,n}(u_{0},v_{1},\ldots,v_{k})). \label{phinz}%
\end{equation}

\end{theorem}

\begin{proof}
Assume that the quantity in (\ref{Enzeta}) is independent of $u_{0}$ for every
$n.$ Then the functions $\phi_{n}$ in (\ref{phinz}) are well defined and if
$H_{n}$ is given by (\ref{nrk1fs}) and $v_{j+1}=h_{n-j}(u_{j},u_{j+1})$ for
$j=0,\ldots,k-1$ in (\ref{nrzetaj}) then%
\begin{align*}
\phi_{n}(H_{n}(u_{0},\ldots,u_{k}))  &  =\phi_{n}(h_{n}(u_{0},u_{1}%
),h_{n-1}(u_{1},u_{2}),\ldots,h_{n-k+1}(u_{k-1},u_{k}))\\
&  =E_{n}(u_{0},\zeta_{1,n}(u_{0},h_{n}(u_{0},u_{1})),\ldots,\\
&  \hspace{0.6in}\zeta_{k,n}(u_{0},h_{n}(u_{0},u_{1}),\ldots,h_{n-k+1}%
(u_{k-1},u_{k})).
\end{align*}

Now, observe that%
\[
\zeta_{1,n}(u_{0},h_{n}(u_{0},u_{1}))=\theta_{n}^{-1}([\mu_{n}(u_{0}%
)]^{-1}\ast\mu_{n}(u_{0})\ast\theta_{n}(u_{1}))=u_{1}.
\]

By way of induction, assume that for $j<k$
\begin{equation}
\zeta_{j,n}(u_{0},h_{n}(u_{0},u_{1}),\ldots,h_{n-j+1}(u_{j-1},u_{j}))=u_{j}
\label{nrzjinduc}%
\end{equation}
and note that%
\begin{align*}
\zeta_{j+1,n}(u_{0},\ldots,h_{n-j}(u_{j},u_{j+1}))  &  =\theta_{n-j}^{-1}%
([\mu_{n-j}(\zeta_{j,n}(u_{0},\ldots,h_{n-j+1}(u_{j-1},u_{j})))]^{-1}\\
&  \hspace{1in}\ast\mu_{n-j}(u_{j})\ast\theta_{n-j}(u_{j+1}))\\
&  =\theta_{n-j}^{-1}([\mu_{n-j}(u_{j})]^{-1}\ast\mu_{n-j}(u_{j})\ast
\theta_{n-j}(u_{j+1}))\\
&  =u_{j+1}.
\end{align*}

It follows that (\ref{nrzjinduc}) is true for all $j=0,1,\ldots,k$ and thus%
\[
\phi_{n}(H_{n}(u_{0},\ldots,u_{k}))=E_{n}(u_{0},\ldots,u_{k})
\]
i.e., $\{H_{n}\}$ as defined by by (\ref{nrk1fs}) is a form symmetry of
Eq.(\ref{nrde}) and therefore, Theorem \ref{nrfacthm} implies the existence of
the associated factorization into equations (\ref{nrk1f}) and (\ref{nrk1cf}).

Conversely, suppose that $\{H_{n}\}$ as defined by by (\ref{nrk1fs}) is a form
symmetry of Eq.(\ref{nrde}). Then there are functions $\phi_{n}$ such that for
all $u_{0},v_{1},\ldots,v_{k}\in G$
\begin{align*}
E_{n}(u_{0},\zeta_{1,n},\ldots,\zeta_{k,n})  &  =\phi_{n}(H_{n}(u_{0}%
,\zeta_{1,n},\ldots,\zeta_{k,n}))\\
&  =\phi_{n}(h_{n}(u_{0},\zeta_{1,n}),\ldots,h_{n-k+1}(\zeta_{k-1,n}%
,\zeta_{k,n}))
\end{align*}
where $\zeta_{j,n}=\zeta_{j,n}(u_{0},v_{1},\ldots,v_{j})$ for $j=1,\ldots,k.$
Since%
\begin{align*}
h_{n-j+1}(\zeta_{j-1,n},\zeta_{j,n})  &  =\mu_{n-j+1}(\zeta_{j-1,n})\ast
\theta_{n-j+1}(\zeta_{j,n})\\
&  =\mu_{n-j+1}(\zeta_{j-1,n})\ast\theta_{n-j+1}(\theta_{n-j+1}^{-1}%
([\mu_{n-j+1}(\zeta_{j-1,n})]^{-1}\ast v_{j}))\\
&  =\mu_{n-j+1}(\zeta_{j-1,n})\ast\lbrack\mu_{n-j+1}(\zeta_{j-1,n})]^{-1}\ast
v_{j}\\
&  =v_{j}%
\end{align*}
it follows that $E_{n}(u_{0},\zeta_{1,n},\ldots,\zeta_{k,n})$ is independent
of $u_{0}$ for all $n,$ as stated.
\end{proof}

Special choices of $\mu_{n}$ and $\theta_{n}$ give analogs of the identity,
inversion and linear form symmetry that are discussed in \cite{hsarx}. The
next example illustrates both Theorem \ref{seminv} and the significant fact
that a \textit{recursive} difference equation may have \textit{non-recursive}
form symmetries.

\begin{example}
Consider the recursive difference equation
\begin{equation}
x_{n}=\sqrt{ax_{n-1}^{2}+bx_{n-1}+cx_{n-2}+d} \label{rnr}%
\end{equation}
where $a,b,c,d\in\mathbb{R}$. To find potential form symmetries of this
equation, first we note that every real solution of (\ref{rnr}) is a
(non-negative) solution of the following quadratic difference equation%
\begin{equation}
x_{n}^{2}-ax_{n-1}^{2}-bx_{n-1}-cx_{n-2}-d=0. \label{rnr1}%
\end{equation}
Based on the existing terms in (\ref{rnr1}), we explore the existence of a
right semi-invertible form symmetry of type%
\begin{equation}
h(u,v)=u^{2}+\alpha v. \label{rnr2}%
\end{equation}
We emphasize that (\ref{rnr2}) is not a recursive form symmetry of the type
discussed in \cite{hstdfs} or \cite{hsarx}. Here $\theta^{-1}(v)=v/\alpha$ so%
\begin{align}
\zeta_{1}(u_{0},v_{1})  &  =\frac{1}{\alpha}(-u_{0}^{2}+v_{1}),\label{rnr3}\\
\zeta_{2}(u_{0},v_{1},v_{2})  &  =\frac{1}{\alpha}v_{2}-\frac{1}{\alpha^{3}%
}(-u_{0}^{2}+v_{1})^{2}. \label{rnr4}%
\end{align}
By Theorem \ref{seminv}, (\ref{rnr2}) is a form symmetry for (\ref{rnr}) if
and only if the following quantity is independent of $u_{0}$%
\begin{equation}
E(u_{0},\zeta_{1},\zeta_{2})=u_{0}^{2}-a\zeta_{1}^{2}-b\zeta_{1}-c\zeta_{2}-d.
\label{rnr5}%
\end{equation}
Using (\ref{rnr3}) and (\ref{rnr4}) in (\ref{rnr5}) and setting the
coefficients of all terms containing $u_{0}$ equal to zero gives the following
two distinct conditions on parameters%
\[
1+\frac{b}{\alpha}=0,\quad-\frac{a}{\alpha^{2}}+\frac{c}{\alpha^{3}}=0.
\]
From these we obtain%
\begin{equation}
\alpha=-b,\quad c=a\alpha=-ab. \label{rnr6}%
\end{equation}
If $b,c\not =0$ then conditions (\ref{rnr6}) imply the existence of a form
symmetry $u^{2}-bv$ for (\ref{rnr}) with a corresponding factorization:%
\begin{align}
t_{n}-at_{n-1}  &  =d,\label{rnr7}\\
x_{n}^{2}-bx_{n-1}  &  =t_{n}. \label{rnr8}%
\end{align}
The positive square root of $x_{n}$ in the cofactor equation (\ref{rnr8}) can
be used to obtain a factorization of the recursive equation (\ref{rnr}) as%
\begin{align*}
t_{n}  &  =at_{n-1}+d,\\
x_{n}  &  =\sqrt{bx_{n-1}+t_{n}},\quad t_{0}=x_{0}^{2}-bx_{-1}.
\end{align*}

Note that this factorization of (\ref{rnr}) as a system of recursive equations
is derived indirectly via the non-recursive equation (\ref{rnr1}) and its
factorization (\ref{rnr7}) and (\ref{rnr8}) rather than directly through a
semiconjugate relation.
\end{example}

The existence and asymptotic behaviors of real solutions discussed in the
preceding example are not as easily inferred from a direct investigation of
(\ref{rnr}). For a discussion of solutions of (\ref{rnr}) using the above
factorization see \cite{hsfsorbk}.

If $\{H_{n}\}$ is a semi-invertible (right and left) form symmetry of
Eq.(\ref{nrde}) then the following result states that the cofactor equation
(\ref{nrk1cf}) can be expressed in recursive form.

\begin{corollary}
Assume that the functions $h_{n}$ in Theorem \ref{seminv} are semi-invertible
so that both $\mu_{n}$ and $\theta_{n}$ are bijections. Then Eq.(\ref{nrde})
has the following factorization%
\begin{gather*}
\phi_{n}(t_{n},\ldots,t_{n-k+1})=\iota\\
x_{n}=\mu_{n}^{-1}(t_{n}\ast\lbrack\theta_{n}(x_{n-1})]^{-1})
\end{gather*}
in which the cofactor equation is recursive.
\end{corollary}
%
%%%%%%%%%%%%%%%%%%%%%%%%%%%%%%%%%%%%%%%%%%%%%%%%%%%
%%%%%%%%%%%%%%%%%%%%%%%%%%%%%%%%%%%%%%%%%%%%%%%%%%%
%
\section{Quadratic difference equations}

In \cite{hstdfs} it is shown that a non-homogeneous and non-autonomous linear
difference equation%
\begin{equation}
x_{n}+a_{1,n}x_{n-1}+\cdots+a_{k,n}x_{n-k}=b_{n},\quad a_{k,n}\not =%
0\ \text{for all }n\geq0 \label{glin}%
\end{equation}
has the linear form symmetry and with the corresponding SC\ factorization over
a non-trivial algebraic field $\mathcal{F}$ if the associated Riccati
equation
\[
\alpha_{n}=a_{0,n}+\frac{a_{1,n}}{\alpha_{n-1}}+\frac{a_{2,n}}{\alpha
_{n-1}\alpha_{n-2}}+\cdots+\frac{a_{k,n}}{\alpha_{n-1}\cdots\alpha_{n-k}}.
\]
of order $k-1$ has a solution $\{\alpha_{n}\}$ in $\mathcal{F}$. It can be
checked that this condition is equivalent to the existence of a
solution of the homogeneous part of (\ref{glin})
\begin{equation}
x_{n}+a_{1,n}x_{n-1}+\cdots+a_{k,n}x_{n-k}=0 \label{glinh}%
\end{equation}
that is never zero. For if $\{y_{n}\}$ is a nonzero solution of (\ref{glinh})
then the ratio sequence $\{y_{n+1}/y_{n}\}$ is a solution of the Riccati equation 
above. These facts lead to a complete analysis of the factorization of
(\ref{glin}) over algebraic fields; see \cite{hsfsorbk} for additional details.
Riccati difference equations have been studied in \cite{kul} (order one) and
\cite{ric2} (order two).

A natural generalization of the linear equation (\ref{glin}) is the
\textit{quadratic difference equation }over a field $\mathcal{F}$%
\begin{equation}
\sum_{i=0}^{k}\sum_{j=i}^{k}a_{i,j,n}x_{n-i}x_{n-j}+\sum_{j=0}^{k}%
b_{j,n}x_{n-j}+c_{n}=0 \label{nrq}%
\end{equation}
that is defined by the general quadratic expression%
\[
E_{n}(u_{0},u_{1},\ldots,u_{k})=\sum_{i=0}^{k}\sum_{j=i}^{k}a_{i,j,n}%
u_{i}u_{j}+\sum_{j=0}^{k}b_{j,n}u_{j}+c_{n}.
\]

Linear equations are obviously special cases of quadratic ones where
$a_{i,j,n}=0$ for all $i,j,n.$ Further, equations of type (\ref{nrq}) also
include the familiar \textit{rational recursive equations} of type%
\begin{equation}
x_{n}=\frac{-\sum_{i=1}^{k}\sum_{j=i}^{k}a_{i,j,n}x_{n-i}x_{n-j}-\sum
_{j=1}^{k}b_{j,n}x_{n-j}-c_{n}}{\sum_{j=1}^{k}a_{0,j,n}x_{n-j}+b_{0,n}}
\label{qlre}%
\end{equation}
\noindent as special cases where
\begin{equation}
a_{0,0,n}=0\text{ for all }n. \label{qlrec}%
\end{equation}

Special cases of (\ref{qlre}) over the field $\mathbb{R}$ of real numbers
include the\ Ladas rational difference equations%
\[
x_{n}=\frac{\alpha+\sum_{j=1}^{k}\beta_{j}x_{n-j}}{A+\sum_{j=1}^{k}%
B_{j}x_{n-j}}%
\]

\noindent as well as familiar quadratic polynomial equations such as the logistic
equation%
\[
x_{n}=ax_{n-1}(1-x_{n-1}),
\]

\noindent the logistic equation with a delay, e.g.,%
\[
x_{n}=x_{n-1}(a-bx_{n-2}-x_{n-1})
\]

\noindent and the Henon difference equation
\[
x_{n}=a+bx_{n-2}-x_{n-1}^{2}.
\]

These rational and polynomial equations have been studied extensively; see, e.g., 
\cite{cl}, \cite{dkmos1}, \cite{dkmos2}, \cite{gl}, \cite{kol}, \cite{kul},
\cite{ord23}, \cite{hsbk}.

We note that all solutions of (\ref{qlre}) are also solutions of the quadratic
(\ref{nrq}) when condition (\ref{qlrec}) is satisfied. The extensive and still
far-from-complete work on rational equations of type (\ref{qlre}) is a clear
indication that unlike the linear case, the existence of a factorization for
Eq.(\ref{nrq}) is not assured and in general, finding any factorization into
lower order equations is a challenging problem.
%
%%%%%%%%%%%%%%%%%%%%%%%%%%%%%%%%%%%%%%%%%%%%%%%%%%%
%
\subsection{Existence of solutions}

When (\ref{qlrec}) holds solutions of (\ref{qlre}) may be recursively
generated. Difficulties arise only when the denominator becomes zero at some
iteration; such singularities often arise from small sets in the state-space
and are discussed in the literature on rational recursive equations.
However, if
\begin{equation}
a_{0,0,n}\not =0\text{ for all }n \label{annt0}%
\end{equation}
then the problem of the existence of solutions for (\ref{nrq}) is entirely
different in nature. If (\ref{annt0}) holds then the issue is not division by
zero but the existence of square roots in the field $\mathcal{F}$. We examine
this issue for the familiar field $\mathbb{R}$ of real numbers. The next
example illustrates a key idea.

\begin{example}
\label{qmtx}Let $a,b,c$ be real numbers such that
\begin{equation}
a\not =0\text{ and }b,c>0. \label{qmtp}%
\end{equation}
and consider the quadratic difference equation%
\begin{equation}
x_{n}^{2}=ax_{n}x_{n-1}+bx_{n-2}^{2}+c. \label{qmt}%
\end{equation}
We solve for the term $x_{n}$ by completing the squares:%
\begin{align*}
x_{n}^{2}-ax_{n}x_{n-1}+\frac{a^{2}}{4}x_{n-1}^{2}  &  =\frac{a^{2}}{4}%
x_{n-1}^{2}+bx_{n-2}^{2}+c\\
\left(  x_{n}-\frac{a}{2}x_{n-1}\right)  ^{2}  &  =\frac{a^{2}}{4}x_{n-1}%
^{2}+bx_{n-2}^{2}+c.
\end{align*}
Now we take the square root which introduces a binary sequence $\{\beta
_{n}\}_{n=1}^{\infty}$ with $\beta_{n}\in\{-1,1\}$ chosen arbitrarily for
every $n$:%
\begin{align}
x_{n}-\frac{a}{2}x_{n-1}  &  =\beta_{n}\sqrt{\frac{a^{2}}{4}x_{n-1}%
^{2}+bx_{n-2}^{2}+c}\nonumber\\
x_{n}  &  =\frac{a}{2}x_{n-1}+\beta_{n}\sqrt{\frac{a^{2}}{4}x_{n-1}%
^{2}+bx_{n-2}^{2}+c} \label{qmt1}%
\end{align}
Under conditions (\ref{qmtp}), for each fixed sequence $\{\beta_{n}%
\}_{n=1}^{\infty}$ every solution of the recursive equation (\ref{qmt1}) with
real initial values is real because the quantity under the square root is
always non-negative. Furthermore, since%
\[
\sqrt{\frac{a^{2}}{4}x_{n-1}^{2}+bx_{n-2}^{2}+c}>\left\vert \frac{a}{2}%
x_{n-1}\right\vert
\]
it follows that for each $n,$
\begin{align*}
x_{n}  &  >0\quad\text{if\quad}\beta_{n}=1,\\
x_{n}  &  <0\quad\text{if\quad}\beta_{n}=-1.
\end{align*}
This sign-switching implies that a significant variety of oscillating
solutions are possible for Eq.(\ref{qmt}) under conditions (\ref{qmtp}).
Indeed, since $\beta_{n}$ is chosen arbitrarily, for every sequence of
positive integers%
\[
\{m_{1},m_{2},m_{3},\ldots\}
\]
there is a solution of (\ref{qmt}) that starts with positive values of $x_{n}$
for $m_{1}$ terms by setting $\beta_{n}=1$ for $1\leq n\leq m_{1}$. Then
$x_{n}<0$ for the next $m_{2}$ terms with $\beta_{n}=-1$ for $n$ in the range
\[
m_{1}+1\leq n\leq m_{1}+m_{2}%
\]
and so on with the sign of $x_{n}$ switching according to the sequence
$\{m_{n}\}_{n=1}^{\infty}.$
\end{example}

Figure \ref{fig:nr-arbit-oscil} illustrates the last part of the above example.

%************ FIGURE **************

\begin{figure}[tbp] % float placement: (h)ere, page (t)op, page (b)ottom, other (p)age
  \centering
  % file name: C:/HS-Math/arXiv/nr-arbit-oscil.bmp
  \includegraphics[bb=0 0 480 271,width=3.92in,height=2.21in,keepaspectratio]{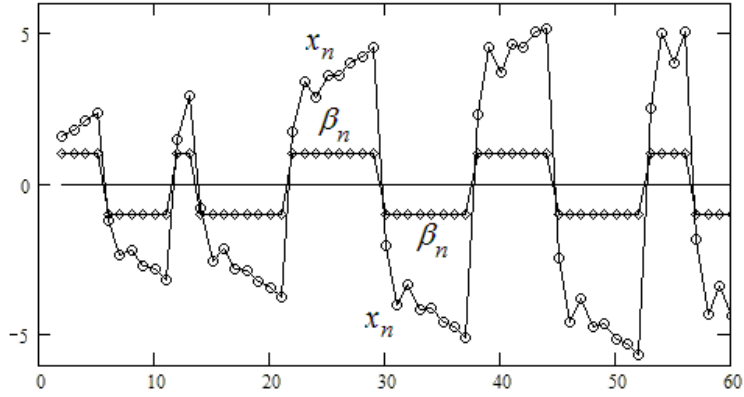}
  \caption{Oscillations in each solution $x_{n}$ exhibit the same pattern as those in the binary sequence $\beta_{n}$ generating it }
  \label{fig:nr-arbit-oscil}
\end{figure}

The method of completing the square discussed in Example \ref{qmtx} can be
applied to \textit{every} quadratic difference equation with real
coefficients. This useful feature enables the calculation of solutions of such
quadratic equations through iteration, a feature that is not shared by
non-recursive difference equations in general. The next result sets the stage
by providing an essential ingredient for the existence theorem.

\begin{lemma}
\label{recls}In the quadratic difference equation%
\begin{equation}
\sum_{i=0}^{k}\sum_{j=i}^{k}a_{i,j,n}x_{n-i}x_{n-j}+\sum_{j=0}^{k}%
b_{j,n}x_{n-j}+c_{n}=0 \label{qslnn}%
\end{equation}
assume that all the coefficients $a_{i,j,n},b_{j,n},c_{n}$ are real numbers
with $a_{0,0,n}\not =0$ for all $n.$ Then $\{x_{n}\}$ is a real solution of
(\ref{qslnn}) if and only if $\{x_{n}\}$ is a real solution of some member of
the family of recursive equations, namely, the recursive class of
(\ref{qslnn})
\begin{equation}
s_{n}=L_{n}(s_{n-1},\ldots,s_{n-k})+\beta_{n}\sqrt{L_{n}^{2}(s_{n-1}%
,\ldots,s_{n-k})+Q_{n}(s_{n-1},\ldots,s_{n-k})} \label{qslnn1}%
\end{equation}
where $\{\beta_{n}\}$ is a fixed but arbitrarily chosen binary sequence with
values in the set $\{-1,1\}$ and for each $n,$
\begin{align}
L_{n}(u_{1},\ldots,u_{k})  &  =\frac{-1}{2a_{0,0,n}}\left(  b_{0,n}+\sum
_{j=1}^{k}a_{0,j,n}u_{j}\right)  ,\label{qslnn2}\\
Q_{n}(u_{1},\ldots,u_{k})  &  =\frac{-1}{a_{0,0,n}}\left(  \sum_{i=1}^{k}%
\sum_{j=i}^{k}a_{i,j,n}u_{i}u_{j}+\sum_{j=1}^{k}b_{j,n}u_{j}+c_{n}\right)  .
\label{qslnn2a}%
\end{align}

\end{lemma}

\begin{proof}
If%
\[
E_{n}(u_{0},u_{1},\ldots,u_{k})=\sum_{i=0}^{k}\sum_{j=i}^{k}a_{i,j,n}%
u_{i}u_{j}+\sum_{j=1}^{k}b_{j,n}u_{j}+c_{n}%
\]

then using definitions (\ref{qslnn2}) and (\ref{qslnn2a}) we may write%
\begin{equation}
E_{n}(u_{0},u_{1},\ldots,u_{k})=a_{0,0,n}[x_{n}^{2}-2L_{n}(u_{1},\ldots
,u_{k})x_{n}-Q_{n}(u_{1},\ldots,u_{k})] \label{qslnn3}%
\end{equation}

Since $a_{0,0,n}\not =0$ for all $n$, the solution set of (\ref{qslnn}) is
identical to the solution set of
\begin{equation}
x_{n}^{2}-2L_{n}(u_{1},\ldots,u_{k})x_{n}-Q_{n}(u_{1},\ldots,u_{k})=0.
\label{qslnn4}%
\end{equation}

Completing the square in (\ref{qslnn4}) gives%
\[
\lbrack x_{n}-L_{n}(u_{1},\ldots,u_{k})]^{2}=L_{n}^{2}(u_{1},\ldots
,u_{k})+Q_{n}(u_{1},\ldots,u_{k}).
\]

which is equivalent to (\ref{qslnn1}).
\end{proof}

Example \ref{qmtx} provides a quick illustration of the preceding result with
\[
L(u_{1},u_{2})=\frac{a}{2}u_{1},\quad Q(u_{1},u_{2})=bu_{2}^{2}+c
\]
where $L_{n}=L$ and $Q_{n}=Q$ are independent of $n.$

We are now ready to present the existence theorem for real solutions of
(\ref{qslnn}). Let $\{\beta_{n}\}$ be a fixed but arbitrarily chosen binary
sequence in $\{-1,1\}$ and define the functions $L_{n}$ and $Q_{n}$ as in
Lemma \ref{recls}. Further, denote the functions on the right hand side of
(\ref{qslnn1}) by $f_{n},$ i.e.,%
\begin{equation}
f_{n}(u_{1},\ldots,u_{k})=L_{n}(u_{1},\ldots,u_{k})+\beta_{n}\sqrt{L_{n}%
^{2}(u_{1},\ldots,u_{k})+Q_{n}(u_{1},\ldots,u_{k})}. \label{fL2Q}%
\end{equation}

These functions on $\mathbb{R}^{k}$ unfold to the self-maps%
\[
F_{n}(u_{1},\ldots,u_{k})=\left(  f_{n}(u_{1},\ldots,u_{k}),u_{1}%
,\ldots,u_{k-1}\right)  .
\]

We emphasize that each function sequence $\{F_{n}\}$ is determined by a given
or \textit{fixed} binary sequence $\{\beta_{n}\}$ as well as the function
sequences $\{L_{n}\}$ and $\{Q_{n}\}$ that are given by (\ref{qslnn}). Clearly
the functions $f_{n}$ are real-valued at a point $(u_{1},\ldots,u_{k}%
)\in\mathbb{R}^{k}$ if and only if%
\begin{equation}
L_{n}^{2}(u_{1},\ldots,u_{k})+Q_{n}(u_{1},\ldots,u_{k})\geq0. \label{L2Q}%
\end{equation}

\begin{theorem}
\label{xrsol}Assume that the following set is nonemtpy:%
\begin{equation}
\mathcal{S}=%
%TCIMACRO{\dbigcap \limits_{n=0}^{\infty}}%
%BeginExpansion
{\displaystyle\bigcap\limits_{n=0}^{\infty}}
%EndExpansion
\left\{  (u_{1},\ldots,u_{k})\in\mathbb{R}^{k}:L_{n}^{2}(u_{1},\ldots
,u_{k})+Q_{n}(u_{1},\ldots,u_{k})\geq0\right\}  . \label{pL2Q}%
\end{equation}
(a) The quadratic difference equation (\ref{qslnn}) has a real solution
$\{x_{n}\}_{n=-k+1}^{\infty}$ if and only if the point $P_{0}=(x_{0}%
,x_{-1},\ldots,x_{-k+1})$ is in $\mathcal{S}$ and there is a binary sequence
$\{\beta_{n}\}$ in $\{-1,1\}$ such that the forward orbit of $P_{0}$ under the
associated maps $\{F_{n}\}$ is contained in $\mathcal{S};$ i.e.,%
\[
\{F_{n}\circ F_{n-1}\circ\cdots\circ F_{1}(P_{0})\}_{n=1}^{\infty}%
\subset\mathcal{S}.
\]
(b) If the maps $\{F_{n}\}$ have a nonempty invariant set $M\subset
\mathcal{S}$ for all $n,$ i.e.,
\[
F_{n}(M)\subset M\subset\mathcal{S}%
\]
then the quadratic difference equation (\ref{qslnn}) has real solutions.
\end{theorem}

\begin{proof}
(a) By Lemma \ref{recls}, $\{x_{n}\}_{n=-k+1}^{\infty}$ is a real solution of
(\ref{qslnn}) if and only if there is a binary sequence $\{\beta_{n}\}$ in
$\{-1,1\}$ such that $\{x_{n}\}_{n=-k+1}^{\infty}$ is a real solution of the
recursive equation (\ref{qslnn1}). Using the notation in (\ref{fL2Q}),
Eq.(\ref{qslnn1}) can be written as%
\begin{equation}
s_{n}=f_{n}(s_{n-1},\ldots,s_{n-k}). \label{sL2Q}%
\end{equation}

Now the forward orbit of the solution $\{x_{n}\}_{n=-k+1}^{\infty}$ of
(\ref{sL2Q}) is the sequence
\[
\mathcal{O}=\{F_{n}\circ F_{n-1}\circ\cdots\circ F_{1}(P_{0})\}_{n=1}^{\infty
}=\{(x_{n-1},\ldots,x_{n-k})\}_{n=1}^{\infty}%
\]

in $\mathbb{R}^{k}$ that starts from $P_{0}\in\mathcal{S}.$ It is clear from
the definition of $T$ that each $x_{n}$ is real if and only if $\mathcal{O}%
\subset\mathcal{S}.$ This observation completes the proof of (a).

(b) Let $P_{0}\in M.$ Then $\{F_{n}\circ F_{n-1}\circ\cdots\circ F_{1}%
(P_{0})\}_{n=1}^{\infty}\subset M\subset\mathcal{S}$ so by (a)
\[
\{(x_{n-1},\ldots,x_{n-k})\}_{n=1}^{\infty}\subset\mathcal{S}.
\]

It follows that each $x_{n}$ is real and thus, $\{x_{n}\}_{n=-k+1}^{\infty}$
of (\ref{sL2Q}). An application of Lemma \ref{recls} now completes the proof.
\end{proof}

If $\mathcal{S}$ has any invariant subset $M=M(\{\beta_{n}\})$ (relative to
some binary sequence $\{\beta_{n}\}$ in $\{-1,1\}$ or equivalently, to some
map sequence $\{F_{n}\}$) then the union of all such invariant sets in
$\mathcal{S}$ is again invariant relative to all relevant binary sequences
$\{\beta_{n}\}$ (or map sequences $\{F_{n}\}$). Invariant sets may exist
(i.e., they are nonempty) for some binary sequences $\{\beta_{n}\}$ in
$\{-1,1\}$ and not others. However, the union of all invariant sets,
\[
\mathcal{M}=%
%TCIMACRO{\dbigcup }%
%BeginExpansion
{\displaystyle\bigcup}
%EndExpansion
\{M(\{\beta_{n}\}):\{\beta_{n}\}\text{ is a binary sequence in }\{-1,1\}\}
\]
is the largest or maximal invariant set in $\mathcal{S}$ and as such,
$\mathcal{M}$ is unique. In particular, if $\mathcal{S}$ is invariant relative
to some binary sequence then $\mathcal{M=S}.$

We refer to $\mathcal{M}$ as the \textit{state-space of real solutions} of
(\ref{qslnn}). The existence of a (nonempty) $\mathcal{M}$ may signal the
occurrence of a variety of solutions for (\ref{qslnn}). In Example \ref{qmtx},
where $\mathcal{S}=\mathbb{R}^{2}$ is trivially invariant (so that
$\mathcal{M=}\mathbb{R}^{2}$) we observed the occurrence of a wide variety of
oscillatory behaviors. Generally, when $\mathcal{S}=\mathbb{R}^{k}$ every
solution of (\ref{qslnn}) with its initial point $P_{0}\in\mathbb{R}^{k}$ is a
real solution. The next result presents sufficient conditions that imply
$\mathcal{S}=\mathbb{R}^{k}$. For cases where $\mathcal{S}$ is a proper subset
of $\mathbb{R}^{k}$\ see \cite{hsfsorbk}.

\begin{corollary}
The state-space of real solutions of (\ref{qslnn}) is $\mathbb{R}^{k}$ if the
following conditions hold for all $n$:%
\begin{align}
A_{j,n}  &  >0,\text{ }B_{i,j,n}=0,\label{rkc0}\\
\sum_{j=1}^{k}\frac{C_{j,n}^{2}}{A_{j,n}}  &  \leq b_{0,n}^{2}-4a_{0,0,n}c_{n}
\label{rkc}%
\end{align}
where for $k\geq2$, $j=1,\ldots,k$, and $n\geq0$%
\begin{align*}
A_{j,n}  &  =a_{0,j,n}^{2}-4a_{0,0,n}a_{j,j,n},\\
B_{i,j,n}  &  =a_{0,i,n}a_{0,j,n}-2a_{0,0,n}a_{i,j,n},\ i<j,\\
C_{j,n}  &  =a_{0,j,n}b_{0,n}-2a_{0,0,n}b_{j,n}.
\end{align*}

\end{corollary}

\begin{proof}
By straightforward calculation the inequality $L_{n}^{2}+Q_{n}\geq0,$ i.e.,
(\ref{L2Q}) is seen to be equivalent to%
\begin{equation}
\sum_{j=1}^{k}A_{j,n}u_{j}^{2}+2\sum_{i=1}^{k-1}\sum_{j=i+1}^{k}B_{i,j,n}%
u_{i}u_{j}+2\sum_{j=1}^{k}C_{j,n}u_{j}+b_{0,n}^{2}\geq4a_{0,0,n}c_{n}
\label{oL2Q}%
\end{equation}
for all $(u_{1},\ldots,u_{k})\in\mathbb{R}^{k}$ with the coefficients
$A_{j,n}$, $B_{i,j,n}$ and $C_{j,n}$ as defined in the statement of the
corollary. By conditions (\ref{rkc}) the double summation term in (\ref{oL2Q})
drops out and we may complete the squares in the remaining terms to obtain the
inequality%
\[
\sum_{j=1}^{k}A_{j,n}\left(  u_{j}+\frac{C_{j,n}}{A_{j,n}}\right)  ^{2}%
\geq\sum_{j=1}^{k}\frac{C_{j,n}^{2}}{A_{j,n}}-b_{0,n}^{2}+4a_{0,0,n}c_{n}.
\]

By (\ref{rkc0}) the left hand side of the above inequality is non-negative
while its right hand side is non-positive so (\ref{oL2Q}) holds under
conditions (\ref{rkc0}) and (\ref{rkc}). The proof is completed by applying
Theorem \ref{xrsol}.
\end{proof}

The next result, which generalizes Example \ref{qmtx}, is obtained by an
immediate application of the above corollary to the non-homogeneous quadratic
equation of order two with constant coefficients.

\begin{corollary}
\label{o2rss}The quadratic difference equation%
\begin{gather*}
x_{n}^{2}+a_{0,1}x_{n}x_{n-1}+a_{0,2}x_{n}x_{n-2}+a_{1,1}x_{n-1}^{2}%
+a_{1,2}x_{n-1}x_{n-2}+\\
\qquad\qquad+a_{2,2}x_{n-2}^{2}+b_{0}x_{n}+b_{1}x_{n-1}+b_{2}x_{n-2}+c_{n}=0
\end{gather*}
has $\mathbb{R}^{2}$ as a state-space of real solutions if the following
conditions are satisfied:%
\begin{align*}
a_{1,1}  &  <\frac{a_{0,1}^{2}}{4},\quad a_{2,2}<\frac{a_{0,2}^{2}}{4},\quad
a_{1,2}=\frac{1}{2}a_{0,1}a_{0,2}\text{ and}\\
c_{n}  &  \leq\frac{b_{0}^{2}}{4}-\frac{(a_{0,1}b_{0}-2b_{1})^{2}}%
{4(a_{0,1}^{2}-4a_{1,1})}-\frac{(a_{0,2}b_{0}-2b_{2})^{2}}{4(a_{0,2}%
^{2}-4a_{2,2})}\text{ for all }n.
\end{align*}

\end{corollary}
%
%%%%%%%%%%%%%%%%%%%%%%%%%%%%%%%%%%%%%%%%%%%%%%%%%%%
%
\subsection{Factorization of quadratic equations}

In the remainder of this paper we investigate conditions for the possible
existence of a special semi-invertible form symmetry for quadratic difference
equations in which both $\mu_{n}$ and $\theta_{n}$ are \textit{linear} maps on
$\mathcal{F}$ for all $n$; hence, this is called a \textit{linear form
symmetry}. This is a first step in a broader study of the latter type of
equation on algebraic fields; see \cite{hsfsorbk} for further discussion.

To simplify the discussion without losing sight of essential ideas, we limit
attention to the case $k=2,$ i.e., the second-order equation
\begin{equation}
E_{n}(x_{n},x_{n-1},x_{n-2})=0 \label{nrqe}%
\end{equation}
on a field $\mathcal{F}$ where%
\begin{align}
E_{n}(u_{0},u_{1},u_{2})  &  =a_{0,0,n}u_{0}^{2}+a_{0,1,n}u_{0}u_{1}%
+a_{0,2,n}u_{0}u_{2}+\label{nrq2}\\
&  \quad a_{1,0,n}u_{1}^{2}+a_{1,2,n}u_{1}u_{2}+a_{2,0,n}u_{2}^{2}+\nonumber\\
&  \qquad b_{0,n}u_{0}+b_{1,n}u_{1}+b_{2,n}u_{2}+c_{n}.\nonumber
\end{align}

In this expression, to assure that $u_{0}$ (corresponding to the $x_{n}$ term)
does \textit{not} drop out, we may assume that for each $n,$%
\begin{equation}
\text{if }a_{0,0,n}=a_{0,1,n}=b_{0,n}=0\text{ or }a_{1,2,n}=a_{2,0,n}%
=b_{2,n}=0,\text{ then }a_{0,2,n}\not =0.\nonumber
\end{equation}

If the quadratic expression (\ref{nrq2}) has a form symmetry $\{H_{n}\}$ where%
\[
H_{n}(u_{0},u_{1},u_{2})=[h_{n}(u_{0},u_{1}),h_{n-1}(u_{1},u_{2})]
\]
then there are sequences $\{\phi_{n}\},\{h_{n}\}$ of functions $\phi_{n}%
,h_{n}:\mathcal{F}^{2}\rightarrow\mathcal{F}$ such that for all $(u_{0}%
,u_{1},u_{2})\in\mathcal{F}^{3},$
\begin{equation}
\phi_{n}(h_{n}(u_{0},u_{1}),h_{n-1}(u_{1},u_{2}))=E_{n}(u_{0},u_{1},u_{2}).
\label{nrfo2}%
\end{equation}

The form symmetry $\{H_{n}\}$ is semi-invertible if there are bijection
$\mu_{n},\theta_{n}:\mathcal{F}\rightarrow\mathcal{F}$ such that
\[
h_{n}(u,v)=\mu_{n}(u)+\theta_{n}(v).
\]

The next corollary of Theorem \ref{seminv} presents conditions that imply the
existence of a linear form symmetry for Eq.(\ref{nrqe}).

\begin{corollary}
\label{qlin}The quadratic difference equation (\ref{nrqe}) has the linear form
symmetry
\begin{equation}
h_{n}(u,v)=\alpha_{n}u+v,\quad\alpha_{n}\not =0\text{ for all }n \label{lnhn}%
\end{equation}
if and only if a sequence $\{\alpha_{n}\}$ exists in the field $\mathcal{F}$
such that all four of the following first-order equations are satisfied:%
\begin{align}
a_{0,0,n}-a_{0,1,n}\alpha_{n}+a_{1,1,n}\alpha_{n}^{2}+a_{0,2,n}\alpha
_{n}\alpha_{n-1}-a_{1,2,n}\alpha_{n}^{2}\alpha_{n-1}+a_{2,2,n}\alpha_{n}^{2}
\alpha_{n-1}^{2} & =0\label{qlin1}\\
a_{0,1,n}-2a_{1,1,n}\alpha_{n}-a_{0,2,n}\alpha_{n-1}+2a_{1,2,n}\alpha
_{n}\alpha_{n-1}-2a_{2,2,n}\alpha_{n}\alpha_{n-1}^{2} & =0\label{qlin2}\\
a_{0,2,n}-a_{1,2,n}\alpha_{n}+2a_{2,2,n}\alpha_{n}\alpha_{n-1} & =0
\label{qlin3}\\
b_{0,n}-b_{1,n}\alpha_{n}+b_{2,n}\alpha_{n}\alpha_{n-1} & =0. \label{qlin4}%
\end{align}
In this case, Eq.(\ref{nrqe}) has a factorization with a first-order factor
equation%
\begin{gather*}
(a_{1,1,n}-a_{1,2,n}\alpha_{n-1}+a_{2,2,n}\alpha_{n-1}^{2})t_{n}^{2}%
+a_{2,2,n}t_{n-1}^{2}+(a_{1,2,n}-2a_{2,2,n}\alpha_{n-1})t_{n}t_{n-1}\\
+(b_{1,n}-b_{2,n}\alpha_{n-1})t_{n}+b_{2,n}t_{n-1}+c_{n}=0.
\end{gather*}
and a cofactor equation $\alpha_{n}x_{n}+x_{n-1}=t_{n}$ also of order one.
\end{corollary}

\begin{proof}
For any nonzero sequence $\{\alpha_{n}\}$ in $\mathcal{F}$ the functions
$h_{n}$ defined by (\ref{lnhn}) are semi-invertible with $\mu_{n}%
(u)=\alpha_{n}u$ and $\theta_{n}(v)=v$ for all $n.$ Note that $\theta_{n}%
^{-1}=\theta_{n}$ for all $n$ and the group structure is the additive group of
the field so the quantities $\zeta_{j,n}$ in Theorem \ref{seminv} take the
forms%
\begin{align*}
\zeta_{1,n}  &  =\zeta_{1,n}(u_{0},v_{1})=-\alpha_{n}u_{0}+v_{1},\\
\zeta_{2,n}  &  =\zeta_{2,n}(u_{0},v_{1},v_{2})=\alpha_{n}\alpha_{n-1}%
u_{0}-\alpha_{n-1}v_{1}+v_{2}.
\end{align*}

By Theorem \ref{seminv}, Eq.(\ref{nrqe}) has the linear form symmetry if and
only if the expression $E_{n}=E_{n}(u_{0},\zeta_{1,n},\zeta_{2,n})$ is
independent of $u_{0}$ for all $n.$ Now%
\begin{align*}
E_{n}  &  =a_{0,0,n}u_{0}^{2}+a_{0,1,n}u_{0}(v_{1}-\alpha_{n}u_{0})+\\
&  a_{0,2,n}u_{0}(\alpha_{n}\alpha_{n-1}u_{0}-\alpha_{n-1}v_{1}+v_{2}%
)+a_{1,1,n}(v_{1}-\alpha_{n}u_{0})^{2}+\\
&  a_{1,2,n}(v_{1}-\alpha_{n}u_{0})(\alpha_{n}\alpha_{n-1}u_{0}-\alpha
_{n-1}v_{1}+v_{2})+\\
&  a_{2,2,n}(\alpha_{n}\alpha_{n-1}u_{0}-\alpha_{n-1}v_{1}+v_{2})^{2}%
+b_{0,n}u_{0}+\\
&  b_{1,n}(v_{1}-\alpha_{n}u_{0})+b_{2,n}(\alpha_{n}\alpha_{n-1}u_{0}%
-\alpha_{n-1}v_{1}+v_{2})+c_{n}.
\end{align*}

Multiplying terms in the above expression gives%
\begin{align*}
E_{n}  &  =a_{0,0,n}u_{0}^{2}+a_{0,1,n}u_{0}v_{1}-a_{0,1,n}\alpha_{n}u_{0}%
^{2}+a_{0,2,n}\alpha_{n}\alpha_{n-1}u_{0}^{2}-\\
&  a_{0,2,n}\alpha_{n-1}u_{0}v_{1}+a_{0,2,n}u_{0}v_{2}+a_{1,1,n}v_{1}%
^{2}-2a_{1,1,n}\alpha_{n}u_{0}v_{1}+\\
&  a_{1,1,n}\alpha_{n}^{2}u_{0}^{2}+a_{1,2,n}\alpha_{n}\alpha_{n-1}u_{0}%
v_{1}-a_{1,2,n}\alpha_{n-1}v_{1}^{2}+a_{1,2,n}v_{1}v_{2}-\\
&  a_{1,2,n}\alpha_{n}^{2}\alpha_{n-1}u_{0}^{2}+a_{1,2,n}\alpha_{n}%
\alpha_{n-1}u_{0}v_{1}-a_{1,2,n}\alpha_{n}u_{0}v_{2}+\\
&  a_{2,2,n}\alpha_{n}^{2}\alpha_{n-1}^{2}u_{0}^{2}-2a_{2,2,n}\alpha_{n}%
\alpha_{n-1}^{2}u_{0}v_{1}+2a_{2,2,n}\alpha_{n}\alpha_{n-1}u_{0}v_{2}+\\
&  a_{2,2,n}(\alpha_{n-1}v_{1}-v_{2})^{2}+b_{0,n}u_{0}+b_{1,n}v_{1}%
-b_{1,n}\alpha_{n}u_{0}+\\
&  b_{2,n}\alpha_{n}\alpha_{n-1}u_{0}-b_{2,n}(\alpha_{n-1}v_{1}-v_{2})+c_{n}.
\end{align*}

Terms containing $u_{0}$ or $u_{0}^{2}$ must sum to zeros. Rearranging terms
in the preceding expression gives%
\begin{align*}
E_{n}  &  =(a_{0,0,n}-a_{0,1,n}\alpha_{n}+a_{0,2,n}\alpha_{n}\alpha
_{n-1}+a_{1,1,n}\alpha_{n}^{2}-a_{1,2,n}\alpha_{n}^{2}\alpha_{n-1}+\\
&  a_{2,2,n}\alpha_{n}^{2}\alpha_{n-1}^{2})u_{0}^{2}+(a_{0,1,n}-a_{0,2,n}%
\alpha_{n-1}-2a_{1,1,n}\alpha_{n}+\\
&  2a_{1,2,n}\alpha_{n}\alpha_{n-1}-2a_{2,2,n}\alpha_{n}\alpha_{n-1}^{2}%
)u_{0}v_{1}+(a_{0,2,n}-a_{1,2,n}\alpha_{n}+\\
&  2a_{2,2,n}\alpha_{n}\alpha_{n-1})u_{0}v_{2}+(b_{0,n}-b_{1,n}\alpha
_{n}+b_{2,n}\alpha_{n}\alpha_{n-1})u_{0}+\\
&  (a_{1,1,n}-a_{1,2,n}\alpha_{n-1}+a_{2,2,n}\alpha_{n-1}^{2})v_{1}%
^{2}+a_{2,2,n}v_{2}^{2}+(a_{1,2,n}-\\
&  2a_{2,2,n}\alpha_{n-1})v_{1}v_{2}+(b_{1,n}-b_{2,n}\alpha_{n-1}%
)v_{1}+b_{2,n}v_{2}+c_{n}.
\end{align*}

Setting the coefficients of variable terms containing $u_{0}$ equal to zeros
gives the four first-order equations (\ref{qlin1})-(\ref{qlin4}). The part of
$E_{n}$ above that does not vanish yields the factor functions
\begin{align*}
\phi_{n}(v_{1},v_{2})  &  =(a_{1,1,n}-a_{1,2,n}\alpha_{n-1}+a_{2,2,n}%
\alpha_{n-1}^{2})v_{1}^{2}+a_{2,2,n}v_{2}^{2}+\\
&  \qquad (a_{1,2,n}-2a_{2,2,n}\alpha_{n-1})v_{1}v_{2}+(b_{1,n}-b_{2,n}\alpha
_{n-1})v_{1}+b_{2,n}v_{2}+c_{n}.
\end{align*}

This expression plus the linear cofactor $t_{n}=h_{n}(x_{n},x_{n-1}%
)=\alpha_{n}x_{n}+x_{n-1}$ give the stated factorization.
\end{proof}

An immediate consequence of Corollary \ref{qlin} is that every second-order,
non-homogeneous linear difference equation%
\[
x_{n}+b_{1,n}x_{n-1}+b_{2,n}x_{n-2}+c_{n}=0
\]
has a linear form symmetry and the corresponding factorization into a pair of
equations of order one (also non-homogeneous, linear) if and only if the
first-order difference equation (\ref{qlin4}) has a solution $\{\alpha_{n}\}$
in $\mathcal{F}\backslash\{0\}$. The existence of a linear form symmetry for
non-homogeneous linear equations of \textit{all orders} can be established by
a calculation similar to that in the proof of Corollary \ref{qlin}. However,
as noted earlier, a complete proof for the general case is already given in
\cite{hstdfs} using the semiconjugate factorization method which applies to
linear equations because they are recursive. Therefore, we need not consider
the linear case any further here.

If all coefficients in $E_{n}(u_{0},u_{1},u_{2})$ are constants except
possibly the free term $c_{n}$ then a simpler version of Corollary \ref{qlin}
is obtained as follows.

\begin{corollary}
\label{cqlin}The quadratic difference equation with constant coefficients%
\begin{gather}
a_{0,0}x_{n}^{2}+a_{0,1}x_{n}x_{n-1}+a_{0,2}x_{n}x_{n-2}+a_{1,1}x_{n-1}%
^{2}+a_{1,2}x_{n-1}x_{n-2}\nonumber\\
+a_{2,2}x_{n-2}^{2}+b_{0}x_{n}+b_{1}x_{n-1}+b_{2}x_{n-2}+c_{n}=0\label{ccqlin}%
\end{gather}
in a non-trivial field $\mathcal{F}$ has the linear form symmetry with
$h(u,v)=\alpha u+v$ if and only if the following polynomials have a common
nonzero root $\alpha$ in $\mathcal{F}$:%
\begin{align}
a_{0,0}-a_{0,1}\alpha+(a_{1,1}+a_{0,2})\alpha^{2}-a_{1,2}\alpha^{3}%
+a_{2,2}\alpha^{4} &  =0,\label{cqlin1}\\
a_{0,1}-(a_{0,2}+2a_{1,1})\alpha+2a_{1,2}\alpha^{2}-2a_{2,2}\alpha^{3} &
=0,\label{cqlin2}\\
a_{0,2}-a_{1,2}\alpha+2a_{2,2}\alpha^{2} &  =0,\label{cqlin3}\\
b_{0}-b_{1}\alpha+b_{2}\alpha^{2} &  =0.\label{cqlin4}%
\end{align}
If such a root $\alpha\not =0$ exists then Eq.(\ref{ccqlin}) has the
factorization%
\begin{gather*}
(a_{1,1}-a_{1,2}\alpha+a_{2,2}\alpha^{2})t_{n}^{2}+(a_{1,2}-2a_{2,2}%
\alpha)t_{n}t_{n-1}+a_{2,2}t_{n-1}^{2}\\
+(b_{1}-b_{2}\alpha)t_{n}+b_{2}t_{n-1}+c_{n}=0,\\
x_{n}=-\frac{1}{\alpha}x_{n-1}+\frac{t_{n}}{\alpha}.
\end{gather*}

\end{corollary}

Equalities (\ref{cqlin1})-(\ref{cqlin4}) often lead to suitable parameter
restrictions implying the existence of a linear form symmetry for a given
difference equation. Here is a sample.

\begin{example}
Consider the following quadratic equation (no linear terms)%
\begin{equation}
x_{n}^{2}+ax_{n}x_{n-1}+bx_{n}x_{n-2}+cx_{n-1}x_{n-2}=\sigma_{n}.\label{qfs}%
\end{equation}
In the absence of linear terms in (\ref{qfs}) equality (\ref{cqlin4}) holds
trivially; the other three equalities (\ref{cqlin1})-(\ref{cqlin3}) take the
following forms%
\begin{align}
1-a\alpha+b\alpha^{2}-c\alpha^{3} &  =0\label{qfs1}\\
a-b\alpha+2c\alpha^{2} &  =0\label{qfs2}\\
b-c\alpha &  =0.\label{qfs3}%
\end{align}
From (\ref{qfs3}) it follows that $\alpha=b/c.$ This nonzero value of $\alpha$
must satisfy the other two equations in the above system so from (\ref{qfs2})
we obtain%
\[
a+\frac{b^{2}}{c}=0
\]
while (\ref{qfs1}) yields%
\[
1-\frac{ab}{c}=0\quad\text{or\quad}c=ab.
\]
Eliminating $b$ and $c$ from the last two equations gives
\[
b=-a^{2},\ c=-a^{3}\ \text{and }\alpha=\frac{1}{a}.
\]
These calculations indicate that the quadratic equation
\[
x_{n}^{2}+ax_{n}x_{n-1}-a^{2}x_{n}x_{n-2}-a^{3}x_{n-1}x_{n-2}=\sigma_{n}%
\]
has the linear form symmetry with the corresponding factorization%
\begin{align*}
a^{2}t_{n}^{2}-a^{3}t_{n}t_{n-1} &  =\sigma_{n},\\
x_{n}+ax_{n-1} &  =at_{n}.
\end{align*}

\end{example}

Corollary \ref{cqlin} also implies the \textit{non-existence} of a linear
form symmetry. The next example illustrates this fact.

\begin{example}
Let $a,b\in\mathbb{C}$ with $b\not =0$ and let $\{c_{n}\}$ be a sequence of
complex numbers. Then the difference equation%
\begin{equation}
x_{n}^{2}+ax_{n-1}^{2}+bx_{n-2}^{2}=c_{n} \label{disk}%
\end{equation}
does not have a linear form symmetry because the equality (\ref{cqlin3}) in
Corollary \ref{cqlin} does not hold for a nonzero complex number $\alpha.$ On
the other hand, the substitution $y_{n}=x_{n}^{2}$ transforms (\ref{disk})
into the non-homogeneous linear equation%
\[
y_{n}+ay_{n-1}+by_{n-2}=c_{n}%
\]
which does have a linear form symmetry in $\mathbb{C}$ (by Corollary
\ref{cqlin}, with equality (\ref{cqlin4}) implying that $-\alpha$ is an
eigenvalue of the homogeneous part). Substituting $x_{n}^{2}$ for $y_{n}$ in
the resulting cofactor equation gives a factorization of (\ref{disk}).
\end{example}

\end{document}